\documentclass[12pt]{article}
\usepackage{algcompatible}
\usepackage{algorithm}

\usepackage[pdftex]{graphicx}
\usepackage{subfigure}
\usepackage{sidecap}
\usepackage{epsfig,amsthm,amsmath,color, amsfonts}
\usepackage{epsfig,color}

\usepackage{framed}
\newtheorem{theorem}{Theorem}[section]
\newtheorem{corollary}[theorem]{Corollary}

\newtheorem{lemma}[theorem]{Lemma}

\theoremstyle{definition}
\theoremstyle{definition}
\theoremstyle{observation}

\newcommand{\comment}[1]{}
\newcommand{\QED}{\mbox{}\hfill \rule{3pt}{8pt}\vspace{10pt}\par}

\newcommand{\pr}{PageRank}

\def\eps{{\epsilon}}

\def\polylog{\operatorname{polylog}}
\newcommand{\ignore}[1]{}
\newcommand{\eat}[1]{}

\usepackage{times}
\usepackage[small,compact]{titlesec}
\usepackage[small,it]{caption}

\newcommand{\squishlist}{
 \begin{list}{$\bullet$}
  { \setlength{\itemsep}{0pt}
     \setlength{\parsep}{3pt}
     \setlength{\topsep}{3pt}
     \setlength{\partopsep}{0pt}
     \setlength{\leftmargin}{1.5em}
     \setlength{\labelwidth}{1em}
     \setlength{\labelsep}{0.5em} } }
\newcommand{\squishend}{
  \end{list}  }

\def\eps{{\epsilon}}

\date{}

\begin{document}

\title{Fast Distributed PageRank Computation \footnote{Appeared in Theoretical Computer Science (TCS), volume 561, pages 113 -121, 2015 \cite{tcs2015}.}}

\author{Atish {Das Sarma} \thanks{eBay Research Labs, eBay Inc., CA, USA.
\hbox{E-mail}:~{\tt atish.dassarma@gmail.com}} \and  Anisur Rahaman Molla \thanks{Division of Mathematical
Sciences, Nanyang Technological University, Singapore 637371. \hbox{E-mail}:~{\tt anisurpm@gmail.com}.} \and Gopal Pandurangan \thanks{Division of Mathematical
Sciences, Nanyang Technological University, Singapore 637371 and Department of Computer Science, Brown University, Providence, RI 02912, USA. \hbox{E-mail}:~{\tt gopalpandurangan@gmail.com}. 
Supported in part by the following research grants: Nanyang Technological University grant M58110000, Singapore Ministry of Education (MOE) Academic Research Fund (AcRF) Tier 2 grant MOE2010-T2-2-082, and a grant from the US-Israel Binational Science Foundation (BSF).
} \and  Eli Upfal \thanks{Department of Computer Science, Brown University, Providence, RI 02912, USA. \hbox{E-mail}:~{\tt eli\_upfal@brown.edu}. Partially supported by NSF BIGDATA Award IIS 1247581.}}

\maketitle \thispagestyle{empty}

\maketitle

\begin{abstract}


Over the last decade, PageRank has gained importance in a wide range of applications and domains, ever since it first proved to be  effective in determining node importance in large graphs (and was a pioneering idea behind Google's search engine). In distributed computing alone, PageRank vector, or more generally random walk based quantities have been used for several different applications ranging from determining important nodes, load balancing, search, and identifying connectivity structures. 
Surprisingly, however, there has been little work towards designing provably efficient fully-distributed algorithms for computing PageRank. The difficulty is that traditional matrix-vector multiplication style iterative methods may not always adapt well to the distributed setting owing to communication bandwidth restrictions and convergence rates. 

In this paper, we present fast random walk-based distributed algorithms  for computing PageRanks in  general  graphs  and prove strong bounds on the round complexity.  We first present a distributed algorithm that  takes $O\big(\log n/\eps \big)$ rounds with high probability on any graph (directed or undirected), where $n$ is the network size and $\eps$ is the reset probability used in the PageRank computation (typically $\eps$ is a fixed constant).  We then present a faster algorithm that takes
$O\big(\sqrt{\log n}/\eps \big)$ rounds in undirected graphs. 
Both of the above algorithms  are  scalable, as each node  sends only small ($\polylog n$) number of  bits over each edge per round. 
To the best of our knowledge, these are the first fully distributed algorithms for computing PageRank vector with provably efficient running time. 
\end{abstract}

\noindent {\bf Keywords:} PageRank, Distributed Algorithm,  Random Walk, Monte Carlo Method

\medskip


\section{Introduction}\label{sec:intro}

In the last decade, PageRank has emerged as a very powerful measure of relative importance of nodes in a network. The term PageRank was first introduced in \cite{anatomy+page98,page99} where it was used to rank the importance of webpages on the Web. Since then, PageRank has found a wide range of applications in a variety of domains within computer science such as distributed networks, data mining, Web algorithms, and distributed computing \cite{pr-survey+05,Bianchini05,cook04,LangvilleM03}. Since PageRank vector or PageRanks is essentially the steady state distribution or the top eigenvector of the Laplacian corresponding to a slightly modified random walk process, it is an easily defined quantity. However, the power and applicability of PageRank arises from its basic intuition of being a way to naturally identify ``important'' nodes, or in certain cases, similarity between nodes. 

While there has been recent work on performing random walks efficiently in distributed networks \cite{ppr-bahmani2010,atish+pods08}, surprisingly, little  provable results are known towards efficient distributed computation of PageRanks. This is perhaps because the traditional method of computing PageRanks is to apply iterative methods i.e., do matrix-vector multiplications till (near)-convergence. Since such techniques may not adapt well in certain settings, when dealing with a global network with only local views (as is common in distributed networks such as Peer-to-Peer (P2P) networks), and particularly, very large networks, it becomes crucial to design far more efficient  techniques. Therefore, PageRank computation using Monte Carlo methods is more appropriate in a distributed model where only messages of limited size are permitted to be sent over each edge in each round.

To elaborate, a naive way to compute PageRank of nodes  in a distributed network is simply scaling iterative PageRank algorithms to
distributed environment. But this is firstly not trivial, and secondly expensive even if doable. As each iteration step needs computation results of previous steps, there needs to be continuous synchronization and several messages may need to be exchanged.  Further, the convergence time may be large. It is important to design efficient and localized distributed algorithms as communication overhead is more important than CPU and memory usage in distributed page ranking. We take all these concerns into consideration and design highly efficient fully decentralized algorithms for computing the PageRank vector in  distributed networks. \\

\noindent{\bf Our Contributions.}
In this paper, to the best of our knowledge, we present the first provably efficient fully decentralized algorithms for estimating PageRanks under a variety of settings. Our algorithms are scalable, since each node  sends only $\polylog n$ bits per round. 
Specifically, our contributions are as follows:
\begin{itemize}
\renewcommand{\labelitemi}{$\bullet$}
\item We present an algorithm, {\sc Basic-PageRank-Algorithm} (cf. Algorithm \ref{alg:simple-pagerank-walk}), that computes  PageRanks accurately in $O\big(\frac{\log n}{\eps} \big)$ rounds with high probability\footnote{Throughout, ``with high probability (w.h.p.)" means with probability at least $1 - 1/n^{c}$, where $n$ is the number of nodes in the network and $c > 1$ is  a suitably chosen constant.}, where $n$ is the number of nodes in the network and $\eps$ is the random reset probability in the PageRank random walk \cite{mcm-avrachenkov,ppr-bahmani2010,atish+pods08}. Our algorithm works for any arbitrary network (directed as well as undirected).\\

\item We present an improved algorithm, called as {\sc Improved-PageRank-Algorithm} (cf. Algorithm \ref{alg:pr-walk-undirected}), that computes  PageRanks accurately in {\em undirected graphs} and terminates with high probability in $O\big(\frac{\sqrt{\log n}}{\eps} \big)$ rounds. We note that though PageRank is usually applied for directed graphs (e.g., for the World
Wide Web), it is sometimes also applied in connection with undirected 
graphs as well \cite{fanchung06,undirected12,IvanG11,PhysRev08,Wang07} and is non-trivial to compute (cf. Section \ref{sec:pagerank-def}). In particular, it can be applied for distributed networks when modeled as undirected graphs (as is typically the case, e.g., in P2P network models).  
\end{itemize}

\noindent We note that the {\sc Improved-PageRank-Algorithm} requires only $O(\log^3 n)$ bits to be sent per round per edge, and  the {\sc Basic-PageRank-Algorithm} requires only $O(\log n)$ bits per round per edge.

 




\section{Background and Related Work}\label{sec:background}
\subsection{Distributed Computing Model}\label{model}


We model the communication network as an unweighted, connected $n$-node graph $G = (V, E)$. Each node has limited initial knowledge. Specifically, we assume that each node is associated with a distinct identity number  (e.g., its IP address). 
At the beginning of the computation, each node $v$ accepts as input its own identity number which is of length
$O(\log n)$ bits and the identity numbers of its neighbors in $G$. The node may also accept some additional inputs as specified by the problem at hand e.g., the number of nodes in the network. 
A node $v$ can communicate with any node $u$ if $v$ knows the id of $u$.\footnote{This is a typical assumption in the context of P2P and overlay  networks, where a node can establish communication with another node if it knows the other node's IP address.  We sometimes  call this {\em direct} communication,
especially when the two nodes are not neighbors in $G$. Note that our algorithm of Section \ref{sec:simple-algo} uses no direct communication between non-neighbors in $G$.} Initially, each node knows only the ids of its neighbors in $G$. We assume that the communication occurs in  synchronous  {\em rounds}, i.e., nodes run at the same processing speed and any message that is sent by some node
$v$ to its neighbors in some round $r$ will be received by the end of round $r$.
In each round, each node is allowed to send a message of size $\polylog n$ bits through each communication link (this applies to both communication via an edge in the network as well as direct communication).  

There are several measures of efficiency of distributed algorithms; here we will focus on the running time, i.e. the number of rounds of distributed communication. Note that the computation that is performed by the nodes locally is ÒfreeÓ, i.e., it does not affect the number of rounds.



\subsection{PageRank}\label{sec:pagerank-def}
We formally define the \pr of a graph $G = (V, E)$. Let $\eps$ be a small constant which is fixed ($\eps$ is called the {\em reset} probability, i.e., with probability $\eps$,  the random walk starts from a node chosen uniformly at random among all nodes in the network). The PageRank vector of a graph (e.g., see \cite{mcm-avrachenkov,ppr-bahmani2010,pr-survey+05,atish+pods08}) is the {\em stationary distribution} vector $\pi$ of the following special type of random walk: at each step of the random walk, with probability $\eps$ the walk starts from a randomly chosen node and with remaining probability $1-\eps$, the walk follows a randomly chosen outgoing (neighbor) edge from the current node and moves to that neighbor.\footnote{We sometime use the terminology ``PageRank random walk'' for this special type of random walk process.} Therefore the PageRank transition matrix on the state space (or vertex set) $V$ can be written as 
\begin{equation}\label{equ:transition-prob}
P = \left(\frac{\eps}{n}\right) J + \big(1-\eps \big)Q
\end{equation} 
where $J$ is the matrix with all entries $1$ and $Q$ is the transition matrix of a simple random walk on $G$ defined as $Q_{ij} = 1/k$, if $j$ is one of the $k>0$ outgoing links of $i$, otherwise $0$. Computing the PageRanks and its variants efficiently in various computation models has been of tremendous research interest in both academia and industry. For a detailed survey of \pr see e.g., \cite{pr-survey+05,LangvilleM03}.  
We note that PageRank is well-defined in both directed and undirected graphs. Note that it is difficult to
compute the PageRank distribution (exactly) analytically   (and no  analytical formulas  are known for general directed graphs) and hence  various computational methods have been 
used to estimate the PageRank distribution. In fact, this is true even for  {\em undirected} graphs as well \cite{undirected12}.

There are mainly two broad approaches to computing PageRanks (e.g., see \cite{bahmani11}). One is to using linear algebraic techniques, (e.g., the Power Iteration \cite{page99}) and the other
approach is Monte Carlo \cite{mcm-avrachenkov}. In the Monte Carlo method, the basic idea is to approximate PageRanks by directly simulating the
corresponding random walk and then estimating the stationary distribution with the performed walk's distribution. In \cite{mcm-avrachenkov} Avrachenkov et al., proposed the following Monte Carlo method
for \pr approximation: Perform $K$ random walks (according to the \pr transition probability) starting from each node $v$ of the graph $G$. For each walk, terminate the walk with its first reset instead of moving to a random node. It is shown that the frequencies of visits of all these random walks to different nodes will approximate the PageRanks. Our distributed algorithms are based on the above  method. 

Monte Carlo methods are efficient, light weight and highly scalable \cite{mcm-avrachenkov}. Monte Carlo methods have been useful  in designing algorithms for PageRank and its variants in  important computational models like data streaming \cite{atish+pods08} and MapReduce \cite{bahmani11}. 
The works in \cite{ssb03,shi03} study distributed implementation of \pr in peer-to-peer networks but use iteration methods.

\section{A Distributed Algorithm for PageRank}\label{sec:simple-algo}

We present a Monte Carlo based distributed algorithm for computing PageRank distribution of a network \cite{mcm-avrachenkov}. The main idea of our algorithm (formal pseudocode is given in Algorithm \ref{alg:simple-pagerank-walk}) is as follows. Perform $K$ ($K$ will be fixed appropriately later) random walks starting from each node of the network in parallel. In each round, each random walk independently  goes to a random (outgoing) neighbor with probability $1-\eps$ and with the remaining probability (i.e., $\eps$) terminates in the current node. Henceforth, we call such a random walk a {\em `PageRank random walk'}. In \cite{mcm-avrachenkov}, this random walk process is shown to be equivalent to one based on the PageRank transition matrix $P$, defined in Section 2.2. It is easy to see that picking each node as starting point for the same number of times (i.e., restarting walks according to the uniform distribution) accounts for  the $(\eps/n) J$ term in Equation~\ref{equ:transition-prob};  and between any two restarts, we just have a simple random walk that terminates with probability $\eps$ in each step --- which accounts for  the $(1-\eps)Q$ term. Since $\eps$ is the probability of termination of a walk in each round, the expected length of every walk is $1/\eps$ and the length  will be at most $O(\log n/\eps)$ with high probability.  Let every node $v$ count the number of visits (say, $\zeta_v$) of all the walks that go through it. Then, after termination of all walks in the network, each node $v$ computes (estimates)  PageRank $\pi_v$ as $\tilde \pi_v = \frac{\zeta_v \eps}{n K}$. Notice that $\frac{nK}{\eps}$ is the (expected) total number of visits over all nodes of all the $n K$ walks. The above idea of counting the number of visits is a standard technique to approximate PageRanks (see e.g., \cite{mcm-avrachenkov,ppr-bahmani2010}). We want to note that our algorithm in this section does not require any direct communication between non-neighbors. 

We show in the next section that the above algorithm computes PageRank vector $\pi$ accurately (with high probability) for an appropriate value of $K$. The main technical challenge in implementing the above method is  that performing many walks from each node in parallel can create a lot of congestion. Our algorithm uses a crucial idea to overcome the congestion. We show that (cf. Lemma \ref{lem:congestion}) that there will be no congestion in the network even if we start a polynomial number of random walks from every  node in parallel. The main idea is based on the Markovian (memoryless) properties of the random walks and the process that terminates the random walks. To calculate how many walks move from  node $i$ to node $j$, node $i$ only needs to know the number of walks that reached it. It does not need to know the sources of these walks or the transitions that they took before reaching node $i$.  Thus it is enough to  send the {\em count} of the number of walks that pass through a node. The algorithm runs till all the walks are terminated which is at most $O(\log n/\eps)$ rounds with high probability. Then every node $v$ outputs  PageRank as the ratio between the number of visits (denoted by $\zeta_v$) to it and the total number of visits over all nodes of all the walks $\big(\frac{nK}{\eps} \big)$.  We show that our algorithm computes approximate PageRanks  in $O(\log n/\eps)$ rounds with high probability (cf. Theorem \ref{thm:main-round}).

\newcommand{\mindegree}[0]{\delta}
\begin{algorithm}[H]
\caption{\sc Basic-PageRank-Algorithm}
\label{alg:simple-pagerank-walk}
\textbf{Input (for every node):} Number of nodes $n$ and reset probability $\eps$.\\
\textbf{Output:} Approximate PageRank of each node.\\
\textbf{[Each node $v$ starts $K = c\log n$ walks, where $c =  \frac{2}{\delta' \eps}$ and $\delta'$ is defined in Section \ref{sec:correctness}. All walks keep moving in parallel until they terminate. The termination probability of each walk is $\eps$, so the expected length of each walk is $1/\eps$.]}\\
\begin{algorithmic}[1]
\STATE Each node $v$ maintains a count variable ``$couponCount_v$" corresponding to number of random walk coupons held by $v$. Initially, $couponCount_v = K$ for starting $K$ random walks. 

\STATE Each node $v$ also maintains a counter $\zeta_v$ for counting the number visits of  random walks to it. Set $\zeta_v = K$. 

\FOR{round $i = 1, 2, \ldots, B\log n/\eps$} \hspace{0.1in}//[for sufficiently large constant $B$]

\STATE Each node $v$ holding at least one alive coupon (i.e., $couponCount_v \neq 0$) does the following in parallel: 

\STATE For every neighbor $u$ of $v$, set $T^u_v = 0$  \hspace{0.5in}// [$T^u_v$ is the number of random walks moving from $v$ to $u$ in round $i$] 

\FOR{$j = 1,2, \ldots, couponCount_v$} 

\STATE With probability $1 - \eps$, pick a uniformly random outgoing neighbor $u$

\STATE $T^u_v : = T^u_v + 1$ 

\ENDFOR

\STATE Send the coupon counter number $T^u_v$ to the respective outgoing neighbors $u$. 

\STATE Each node $u$ computes:  $\zeta_u = \zeta_u + \sum_{v \in N(u)} T^u_v$.  \hspace{0.5in} //[the quantity $\sum_{v \in N(u)} T^u_v$ is  the total number of visits of random walks to $u$ in $i$-th round (from its neighbors)]

\STATE Each node $u$ update the count variable $couponCount_u = \sum_{v \in N(u)} T^u_v$

\ENDFOR

\STATE Each node $v$ outputs its PageRank as $\frac{\zeta_v \eps}{c n \log n}$.

\end{algorithmic}

\end{algorithm}

\subsection{Analysis}
Our algorithm computes the PageRank of each node $v$ as $\tilde \pi_v = \frac{\zeta_v \eps}{n K}$ and we say that $\tilde \pi_v$ approximates original PageRank $\pi_v$. We first focus on the correctness of our approach and then analyze the running time. 

\subsection{Correctness of PageRank Approximation}\label{sec:correctness}
The correctness of the above approximation follows directly from the main result of \cite{mcm-avrachenkov} (see Algorithm $4$ and Theorem $1$) and also from \cite{ppr-bahmani2010} (Theorem $1$). In particular, it is mentioned in \cite{mcm-avrachenkov,ppr-bahmani2010} that the approximate \pr value is quite good even for $K = 1$. It is easy to see that the expected value of $\tilde \pi_v$ is $\pi_v$ (formal proof is given in \cite{mcm-avrachenkov}). Now it follows from the Theorem~1 in \cite{ppr-bahmani2010} that, $\tilde \pi_v$ is sharply concentrated around its expectation $\pi_v$.
We included the proof of the theorem below for the sake of completeness.


\begin{theorem}[Theorem~1 in \cite{ppr-bahmani2010}]\label{thm:pr-concentration-bahmani}
 $\Pr \big[\mid \tilde \pi_v - \pi_v \mid \geq \delta \pi_v \big] \leq e^{-nK\pi_v \delta'}$, where $\delta'$ is a constant depending on $\eps$, the reset probability and $\delta$. 
\end{theorem}
\begin{proof}
For simplicity we first show the result assuming $K = 1$. For general value of $K$, it will follow in the similar way. Fix an arbitrary node $v$. Define $X_u$ to be $\eps$ times
the number of visits to $v$ in the walk started at $u$, $Y_u$ to be
the length of this walk, $W_u = \eps Y_u$, and $x_u = E[X_u]$. Then, $X_u$'s are independent, $\tilde \pi_v = \frac{\sum_u X_u}{n}$ and hence $\pi_v = \frac{\sum_u x_u}{n}$, $0 \leq X_u \leq W_u$, and $E[W_u] = 1$. Then it follows easily that, 
\begin{align*}
E\big[e^{tX_u}\big] & \leq x_u E\big[e^{tW_u}\big] + 1 - x_u  \hspace{0.2in}\text{[From the definition of expectation]} \\ & = x_u \left(E\big[e^{tW_u}\big] - 1\right) + 1 \\ & \leq e^{- x_u \left(1 - E\big[e^{tW_u}\big]\right)} \hspace{0.5in} \text{[Since $1 + y \leq e^y$ for any $y$]}
\end{align*}
Thus,
\begin{align*}
\Pr \big[\tilde \pi_v \geq (1 + \delta)\pi_v \big]  & \leq \frac{E[e^{tn\tilde \pi_v}]}{e^{tn(1 + \delta)\pi_v}} \hspace{0.5in} \text{[Markov's inequality]} \\ & = \frac{E\big[e^{t\sum_u X_u}\big]}{e^{tn(1 + \delta)\pi_v}} = \frac{\prod_u E[e^{tX_u}]}{e^{tn(1 + \delta)\pi_v}} \leq \frac{\prod_u e^{-x_u \left(1 - E[e^{tW_u}]\right)}}{e^{tn(1 + \delta)\pi_v}} \\ & = \frac{e^{-\left(\sum_u x_u \left(1 - E[e^{tW_u}]\right)\right)}}{e^{tn(1 + \delta)\pi_v}} = \frac{e^{-n\pi_v \left(1 - E[e^{tW}] \right)}}{e^{tn(1 + \delta)\pi_v}}  \\ & = e^{-n\pi_v \left(1 + t(1 + \delta) - E[e^{tW}]\right)} \\ & \leq e^{-n\pi_v \delta'}
\end{align*}
where $W = \eps Y$ is a random variable with $Y$ having geometric distribution with parameter $\eps$, and $\delta' = 1 + t(1 + \delta) - E[e^{tW}]$ is a constant depending on $\delta$ and $\eps$, and can be found by optimization over $t$. 

The proof for the other direction $\Pr[\tilde \pi_v \leq (1 - \delta)\pi_v]$ is similar. 
\end{proof}
From the above bound (cf. Theorem \ref{thm:pr-concentration-bahmani}), we see that for $K = \frac{2\log n}{\delta' n\pi_{min}}$, $\Pr[\mid \tilde \pi_v - \pi_v \mid \geq \delta \pi_v] \leq n^{-2}$ for any $v$, where $\pi_{min}$ is minimal PageRank. Using union bound, it follows that there exist a node $v$ such that $\Pr[\mid \tilde \pi_v - \pi_v \mid \geq \delta \pi_v]$ is at most $|V|n^{-2} = 1/n$. Hence, for all nodes $v$, $\mid \tilde \pi_v - \pi_v \mid \leq \delta \pi_v$ with probability at least $1 - 1/n$, i.e., with high probability. This implies that we get a  $\delta$-approximation of the PageRank vector with high probability for $K = \frac{2\log n}{\delta' n\pi_{min}}$. Note that $\delta$ can be arbitrary.  
Since the \pr of any node is at least $\eps/n$ (i.e., the minimal \pr value, $\pi_{min} \geq \eps/n$), so it gives $K = \frac{2\log n}{\delta' \eps}$. 
For simplicity we define that $c =  \frac{2}{\delta' \eps}$, which is constant assuming $\delta$ (and hence $\delta'$) and $\eps$ are constant. Therefore, it is enough if we perform $c\log n$ \pr random walks from each  node. We note that while this  value of $K$ is sufficient to guarantee a constant approximation of the PageRanks, our algorithm permits a larger value of $K$, allowing for tighter approximation with the same running time (follows from Lemma \ref{lem:congestion} below). Now we focus on the running time of our algorithm. 

\subsection{Time Complexity}\label{sec:complexity}
From the above section we see that our algorithm is able to compute the PageRank vector $\pi$ in $O(\log n/\eps)$ rounds with high probability if we can perform $c\log n$ walks from each node in parallel without any congestion.  The lemma below guarantees that there will be no congestion even if we do a polynomial number of walks in parallel.   


\begin{lemma}\label{lem:congestion}
The algorithm can be implemented such that the message size is at most $O(\log n)$ per each edge in every round. 
\end{lemma}
\begin{proof}
It follows from our algorithm that each node only needs to count the number of visits of random walks to itself.
Since the total number of random walk coupons in the network is polynomially bounded, $O(\log n)$ bits suffice. 
\end{proof}

\begin{theorem}\label{thm:main-round}
The algorithm {\sc Basic-PageRank-Algorithm} (cf. Algorithm \ref{alg:simple-pagerank-walk}) computes a $\delta$-approximation of the PageRanks in $O\left(\frac{\log n}{\eps} \right)$ rounds with high probability for any constant $\delta$. 
\end{theorem}
\begin{proof}
The algorithm outputs the RageRanks when all the walks terminate. Since the termination probability is $\eps$, in expectation after $1/\eps$ steps, a walk terminates and with high probability (via a Chernoff bound) the walk terminates in $O(\log n/\eps)$ rounds. By the union bound \cite{MU-book-05}, all walks (they are only polynomially many) terminate
in $O(\log n/\eps)$ rounds with high probability. Since all the walks are moving in parallel and there is no congestion (follows from the Lemma \ref{lem:congestion}), all the walks in the network terminate in $O(\log n/\eps)$ rounds with high probability. Hence the algorithm is able to output the PageRanks in $O(\log n/\eps)$ rounds with high probability. The correctness of the PageRanks approximation follows  from \cite{mcm-avrachenkov,ppr-bahmani2010} as discussed earlier in Section \ref{sec:correctness}. The $\delta$-approximation guarantee is follows from the Theorem~\ref{thm:pr-concentration-bahmani}.
\end{proof}

\section{A Faster Distributed \pr Algorithm (for  Undirected Graphs)}\label{sec:undirected}
We present a faster algorithm for PageRanks computation in {\em undirected} graphs. Our algorithm's time complexity holds in the bandwidth restricted communication model, requires only $O(\log^3 n)$ bits to be sent over each link in each round.  

We use a similar Monte Carlo method as described in Section \ref{sec:simple-algo} to estimate PageRanks. This says that the PageRank of a node $v$ is the ratio between the number of visits of \pr random walks to $v$ itself and the sum of all the visits over all nodes in the network. In the previous section (cf. Section \ref{sec:simple-algo}) we show that in $O(\log n/\eps)$ rounds, one can approximate RageRank accurately by walking in a naive way in general graphs. We now outline how to speed up our previous algorithm (cf. Algorithm \ref{alg:simple-pagerank-walk}) using an  idea similar to the one used in \cite{DasSarmaNPT10}. In \cite{DasSarmaNPT10}, it is shown how one can perform {\em a}  simple  random walk in an undirected graph\footnote{In each step, an edge is taken from
the current node $x$ with probability proportional to $1/d(x)$ where
$d(x)$ is the degree of $x$.}  of length $L$  in $\tilde O\big(\sqrt{LD}\big)$ rounds w.h.p. ($D$ is the diameter of the network). The high level idea of their algorithm  is to perform `many' short walks in parallel and later `stitch' them to get the desired longer length walk. To apply this idea in our case, we  modify our approach accordingly as speeding up ({\em many}) PageRank random walks is  different from speeding up {\em one}  simple random walk. 
We show that our improved algorithm (cf. Algorithm \ref{alg:pr-walk-undirected}) approximates PageRanks in $O\big(\frac{\sqrt{\log n}}{\eps}\big)$ rounds. 



\subsection{Description of Our Algorithm}


In  Section \ref{sec:simple-algo}, we showed that by performing $\Theta(\log n)$ walks (in particular we are performing $c\log n$ walks, where $c = \frac{2}{\delta' \eps}$, $\delta'$ is defined in Section \ref{sec:correctness}) of length $\log n/\eps$ from each node, one can estimate the \pr vector $\pi$ accurately with high probability. In this section we focus on the problem of  performing efficiently $\Theta(n\log n)$ walks ($\Theta(\log n)$ from each node) each of length $\log n/\eps$ and count the number of visits of these walks to different nodes.
Throughout, by ``random walk'' we mean the ``PageRank random walk" (cf. Section \ref{sec:simple-algo}). 

The main idea of our algorithm is to first perform `many' short random walks in parallel and then `stitch' those short walks to get the longer walk of length $\log n/\eps$ and subsequently `count' the number of visits of these random walks to different nodes. In particular, our algorithm runs in three phases. In the first phase, each node $v$ performs $d(v) \eta$ ($d(v)$ is degree of $v$) independent `short' random walks of length $\lambda$ in parallel. While value of the parameters $\eta$ and $\lambda$ will be fixed later in the analysis, the assigned value will be $O(\log^2 n/\eps)$ and $\sqrt{\log n}$ respectively. This is done naively by forwarding $d(v)\eta$ `coupons' having the ID of $v$ from $v$ (for each node $v$) for $\lambda$ steps via random walks. Besides the node's ID, we also assign a coupon number ``$Coupon_{ID}$" to each coupon to keep track the path followed by the random walk coupon. The intuition behind performing $d(v)\eta$ short walks is that the PageRanks of an undirected graph is proportional to the degree distribution \cite{undirected12}. Therefore we can easily bound the number of visits of random walks to any node $v$ (cf. Lemma \ref{lem:visit-bound}). At the end of this phase, if node $u$ has $k$ random walk coupons with the ID of a node $v$, then $u$ is a destination of $k$ walks starting at $v$. Note that just after this phase, $v$ has no knowledge of the destinations of its own walks, but it can be  known by direct communication from the destination nodes. The destination nodes (at most $d(v)\eta$) have the ID of the source node $v$. So they can contact the source node via {\em direct} communication. We show that this takes at most constant number of rounds  as only polylogarithmic number of bits are sent (since $\eta$ will be at most $O(\log^2 n/\eps)$).  It is shown that the first phase takes $O\big(\frac{\lambda}{\eps} \big)$ rounds (cf. Lemma \ref{lem:phase1}).

In the second phase, starting at source node $s$, we `stitch' some of the $\lambda$-length walks prepared in first phase. Note that we do this for every node $v$ in parallel as we want to perform $\Theta(\log n)$ walks from each node. The algorithm starts from $s$ and samples one coupon distributed from $s$ in Phase 1. 
In the end of Phase 1, each node $v$ knows the destination node's ID of its $d(v)\eta$ short walks (or coupons). When a coupon needs to be sampled, node $s$ chooses a coupon number sequentially (in order of the coupon IDs) from the unused set of coupons and informs the destination node (which will be the next stitching point) holding the coupon $C$ by direct communication, since $s$ knows the ID of the destination node at the end of the first phase. 
Let $C$ be the sampled coupon and $v$ be the destination node of $C$. The source $s$ then sends a `token' to $v$ and $s$ deletes the coupon $C$ so that $C$ will not be sampled again next time at $s$. This is because our goal is to produce independent random walks of a given length, so naturally we do not reuse the same short walks, or in other words, this will preserve randomness when we concatenate short walks. The process then repeats. That is, the node $v$ currently holding the token samples one of the coupons it distributed in Phase 1 and forwards the token to the destination of the sampled coupon, say $u$. Nodes $v, u$ are called `connectors' --- they are the endpoints of the short walks that are stitched. A crucial
observation is that the walk of length $\lambda$ used to distribute the corresponding coupons from $s$ to $v$ and from $v$ to $u$ are independent random walks. Therefore, we can stitch them to get a random walk of length $2\lambda$. We therefore can generate a random walk of length $3\lambda, 4\lambda, \ldots$ by repeating this process. We do this until we have completed  a length of at least $\big( O(\log n/\eps) - \lambda \big)$. Then, we complete the rest of the walk by doing the naive random walk algorithm. Note that in the beginning of Phase~2, we first

\begin{algorithm}[H]\small 
\caption{\sc Improved-PageRank-Algorithm}
\label{alg:pr-walk-undirected}
\textbf{Input (for every node):}  Number of nodes $n$, reset probability $\eps$ and short walk length $\lambda = \sqrt{\log n}$.\\
\textbf{Output:} Approximate PageRank of each node.\\


\textbf{Phase 1: (Each node $v$ performs $d(v)\eta = O(d(v)\log^2 n/\eps)$
random walks of length $\lambda = \sqrt{\log n}$. 
At the end of this phase, there are $d(v)\log^2 n/\eps$ (not necessarily
distinct) nodes holding a `coupon' containing the ID of $v$.)}
\begin{algorithmic}[1]

\STATE Each node $v$ construct $Bd(v)\log^2 n/\eps$ messages $C = \langle ID_v,  \lambda, Coupon_{ID} \rangle$. \hspace{0.1in}// [We will refer to these messages created by node $v$ as `coupons created by $v$'.]

\FOR{$i=1$ to $\lambda$}

\STATE This is the $i$-th iteration. Each node $v$ holding at least one coupon does the following in parallel: 
\FOR{each coupon $C$ held by $v$} \hspace{0.3in}// [i.e., the coupons which received by $v$ in the $(i - 1)$-th iteration.]  
\STATE Generate a random number $r \in [0, 1]$. 

\IF {$r< \eps$}

\STATE Terminate the coupon $C$ and keep the coupon as then $v$ itself is the destination. 

\ELSE 

\STATE pick a neighbor $u$ uniformly at random for the coupon $C$ and forward $C$ to $u$. 


\ENDIF
\ENDFOR

\noindent \COMMENT{Note that an iteration could require more than 1 round, because of congestion}

\ENDFOR


\STATE Each destination node sends its ID to the source node, as it has the source node's ID now. \hspace{0.2in}// [destination nodes hold the short walk coupon(s) $C$ and contact the source nodes through {\em direct} communication.]

\end{algorithmic}

\textbf{Phase 2: (Stitch short walks by token forwarding. Stitch approximately $\Theta(\sqrt{\log n}/\eps)$ walks, each of
length  $\sqrt{\log n}$. Recall that each node wants to perform $K = c\log n$ long random walks, where $c =  \frac{2}{\delta' \eps}$ and $\delta'$ is defined in Section \ref{sec:correctness})}



\begin{algorithmic}[1]

\STATE  Each node $v$ generates $K$ ``tokens''  $\langle ID_v,  L \rangle$, where $L$ is a random integer value $x$ chosen with probability $\eps(1-\eps)^{x-1}$ \hspace{0.3in}//  [$L$ is drawn from the geometric distribution with parameter $\eps$ i.e., from the distribution of the lengths of random walks.]

\FOR{$i = 1, 2, \ldots, B_1\sqrt{\log n}/\eps$} \hspace{0.1in}//[for sufficiently large constant $B_1$]

\STATE Each node $v$ holding at least one {\em token} with $L>0$ does the following in parallel: 

\STATE For each token $\langle ID_v,  L \rangle$ with $L\geq \lambda$, send $\langle ID_v,  L - \lambda, Coupon_{ID} \rangle$ to $u$, where $u$ is sampled using a coupon of sequence number $Coupon_{ID}$ from the set of the coupons distributed by $v$ in Phase~1,  and delete the token  $\langle ID_v,  L \rangle$ \hspace{0.5in}//  [$v$ sends to $u$ via the {\em direct} communication.] 

\STATE For each such received message $\langle ID_v,  L - \lambda, Coupon_{ID} \rangle$, node $u$ memorizes $(ID_v, Coupon_{ID})$ and creates a token $\langle ID_u,  L - \lambda \rangle$    \hspace{0.3in}//  [Each node $u$ memorizes it for backtracking in Phase 3.] 

\ENDFOR

\STATE \label{step:count} For the remaining tokens  $\langle ID_v,  L \rangle$ (whose $L >0$), it holds that $L < \lambda$:  for each of them walk naively in parallel for another $\lambda$ steps. 

\end{algorithmic}

\textbf{Phase 3: (Counting the number of visits of  short walks to a node)}
\begin{algorithmic}[1]

\STATE Each node $w$ maintains a counter $\zeta_w$ to keep track of the number of visits of walks at $w$. $\zeta_w$ is initialized to $K$.  

\STATE Each node $u$ which memorizes coupon IDs $(ID_v, Coupon_{ID})$ in Phase 2, does the following in parallel: 


\STATE  For each such coupon, starting from $u$ trace all the short random walks in reverse.

\STATE Count the number of visits to any node $w$ during this reverse tracing and add to $\zeta_w$. Also count the visits during `naively walking' walks (Step \ref{step:count} in Phase 2) and add it to $\zeta_w$.    


\STATE Each node $v$ outputs its PageRank $\pi_v$ as $\frac{\zeta_v \eps}{c n \log n}$. 
  
\end{algorithmic}

\end{algorithm}

\noindent check the length of survival of each walk and then stitch them accordingly.  We show that  Phase 2 finishes in $O\big(\frac{\log n}{\lambda \eps} + \lambda \big)$ rounds (cf. Lemma \ref{lem:phase2}).

In the third phase we count the number of visits of all the random walks to a node. As we have discussed,  we have to create many short walks of length $\lambda$ from each node. Some short walks may not be used to make the long walk of length $\log n/\eps$. We show a technique to count all the used short walks' visits to different nodes.  We note  that after completion of Phase 2, all the $\Theta(n\log n)$  long walks ($\Theta(\log n)$ from each node) have been  stitched.  During stitching  (i.e., in Phase 2), each connector node (which is also the end point of the short walk) should remember the source node and the $Coupon_{ID}$ of the short walk. Then start from  each of the connector nodes and do a walk in reverse  direction (i.e., retrace the short walk backwards) to the respective source nodes in parallel. During the reverse walk, simply count the visits to nodes. It is easy to see that this will take at most $O(\frac{\lambda}{\eps})$ rounds, in accordance with Phase~1 (cf. Lemma \ref{lem:phase3}). Now we analyze the running time of our algorithm {\sc  Improved-PageRank-Algorithm}. The compact pseudo code is given in Algorithm \ref{alg:pr-walk-undirected}.

\subsection{Analysis}
First we are interested in the value of $\eta$ i.e., the number of coupons (short walks)  needed from each node to successfully answer all the stitching requests. Notice that it is possible that $d(v)\eta$ coupons are not enough if $\eta$ is not chosen suitably large: We might forward the token to some node $v$ many times in Phase 2 and all coupons distributed by $v$ in the first phase may be deleted. In other words, $v$ is chosen as a connector node many times, and all its coupons have been exhausted.
If this happens then the stitching process cannot progress. To fix this problem, we use an easy upper bound of the number of visits to any node $v$ of a random walk of length $\ell$ in an undirected graph: $d(v)\ell$ times. 
Therefore each node $v$ will be visited as a connector node at most $O(d(v)\ell)$ times. This implies that each node does not have to prepare too many short walks. 


The following lemma bounds the number of visits to every node when we do $\Theta(\log n)$ walks
from each node, each of length $\log n/\eps$  (note that this is the maximum length of a long walk, w.h.p.). 


\begin{lemma}\label{lem:visit-bound}
If {\em each} node performs $\Theta(\log n)$ random walks of length $\log n/\eps$, then no node $v$ is visited more than $O\big(\frac{d(v) \log^2 n}{\eps} \big)$ times with high probability. 
\end{lemma}
\begin{proof}
We show the above  bound on the number of visits still holds if each node $v$ performs $\Theta \big(d(v)\log n \big)$ random walks of length $\log n/\eps$. Suppose each node $v$ starts $\Theta \big(d(v) \log n \big)$ simple random walks in parallel. We claim that after any given number of steps $i$, the expected number of random walks at node $v$ is still $\Theta \big(d(v)\log n \big)$. Consider the random walk's transition probability matrix $A$. Then, $A{\bf x} = {\bf x}$ holds for the stationary distribution ${\bf x}$ having value $\frac{d(v)}{2m}$, where $m$ is the number of edges in the graph. Now the number of random walks started at any node $v$ is proportional to its stationary distribution, therefore, in expectation, the number of random walks at any node after $i$ steps remains the same. We show this is true with high probability using Chernoff bound technique, since the random walks are independent. For each random walk coupon $C$, any $i = 1, 2, \ldots, \log n/\eps$, and any vertex $v$, we define $W_C^i(v)$ to be the random variable having value
1 if the random walk $C$ is at $v$ after $i^{th}$ step. Let $W^i(v)=\sum_{C: \text{random walk}} W_C^i(v)$, i.e., $W^i(v)$ is the total number of random walks are at $v$ after $i^{th}$ step. By Chernoff bound, for
any vertex $v$ and any $i$,
$$\Pr \big[W^i(v)\geq 18 d(v)\log{n} \big]\leq 2^{-3 d(v)\log{n}} \leq n^{-3}.$$
It follows that the probability that there exists an vertex $v$ and an
integer $1\leq i\leq \log n/\eps$ such that $W^i(v)\geq 18 d(v)\log{n}$ is
at most $|V(G)| (\log n/\eps) n^{-3}\leq \frac{1}{n}$ since $|V(G)| =
n$ and $\log n/\eps \leq n$. Therefore, $W^i(v) \leq 18 d(v)\log{n}$ for all $v$ and for all $i$, with high probability.  

Now, if each node starts $\Theta(\log n)$ independent random walks that terminate with probability $\eps$ in each step, the number of random walks to any node $v$ is dominated from above by $\Theta \big(d(v)\log n \big)$. This is because there will be at most $n\log n$ random walk coupons in the network in each step. Therefore, the total number of visits by all random walks to any node $v$ is bounded by $O\big(d(v) \log^2 n/\eps \big)$ w.h.p., since there are total of $\log n/\eps$ steps.
%
 \end{proof}

It is now clear from the above lemma (cf. Lemma~\ref{lem:visit-bound}) that $\eta = O\big(\log^2 n/\eps \big)$ i.e., each node $v$ has to prepare $O\big(d(v)\log^2 n/\eps \big)$ short walks of length $\lambda$ in Phase 1. Now we show the running time of our algorithm (cf. Algorithm \ref{alg:pr-walk-undirected}) using the following lemmas.  
\begin{lemma}\label{lem:phase1}
Phase 1 finishes in $O\left(\frac{\lambda}{\eps} \right)$ rounds.  
\end{lemma} 
\begin{proof}
It is known from  Lemma \ref{lem:visit-bound} that in Phase 1, each node $v$ performs $O\big(d(v)\log^2 n/\eps \big)$ walks of length $\lambda$. 
Assume that initially each node $v$ starts with $d(v)\log^2 n/\eps$ coupons (or messages) and each coupon takes a random walk according to the \pr transition probability. Now, in the similar way we showed in Lemma \ref{lem:visit-bound} that after any given number of steps $j$ $(1 \leq j \leq \lambda)$, the expected number of coupons at any node $v$ is $d(v)\log^2 n/\eps$. Therefore, in expectation the number of messages, say $X$, that want to go through an edge in any round is at most $2 \log^2 n/\eps$ (from the two end points of the edge). This is because the number of messages, the edge receives from its one end node, say $u$, in expectation is exactly the
number of messages at $u$ divided by $d(u)$. Using Chernoff bound we get, $\Pr[X\geq 24 \log^2 n/\eps] \leq 2^{-4 \log^2 n/\eps} \leq n^{-4}$. By union bound we get that  there exists an edge and an
integer $1\leq j\leq \lambda$ such that the probability of $X \geq 24\log^2 n/\eps$ is
at most $|E(G)| \lambda n^{-4}\leq \frac{1}{n}$, since $|E(G)|\leq n^2$ and $\lambda < n$. Hence the number of messages that go through any edge in any round is at most $24 \log^2 n/\eps = O(\log^2 n/\eps)$ with high probability. So the message size will be at most $O(\log^3 n/\eps)$ bits w.h.p. over any edge in each round (a message contains source IDs and coupon IDs each of which can be encoded using $\log n$ bits). Since our considered model allows polylogarithmic (i.e., $O(\log^3 n)$) bits messages per edge per round, we can extend all the random walk's length from $i$ to length $i+1$ in $O(1/\eps)$ rounds. Therefore, for walks of length $\lambda$ it takes $O(\lambda/\eps)$ rounds as claimed.   
\end{proof}

\begin{lemma}\label{lem:sample-coupon}
With the message size $O(\log n)$ in Phase 2, one stitching step from each node in parallel can be done in one round.  
\end{lemma} 
\begin{proof}
 Each node knows all of its short walks' (or coupons') destination address and the $Coupon_{ID}$. Each time when a source or connector node wants to stitch, it chooses its unused coupons (created in Phase 1) sequentially in order of the coupon IDs. Then it contacts the destination node (holding the coupon) through {\em direct} communication and informs the destination node as the next connector node or stitching point. Therefore, in each round, it is sufficient for any node to send to connector node $u$ the maximal $Coupon_{ID}$ with destination $u$ that it has used so far. This implies that message size of $O(\log n)$ bits per edge suffices for this process. Since we assume the network allows $O(\log^3 n)$ congestion, this one time stitching from each node in parallel will finish in one round.
\end{proof}

\begin{lemma}\label{lem:phase2}
Phase 2 finishes in $O\left(\frac{\log n}{\lambda \eps} + \lambda \right)$ rounds.  
\end{lemma} 
\begin{proof}
Phase 2 is for stitching short walks of length $\lambda$ to get a long walk of length $B_1\log n/\eps$, where the constant $B_1$ is chosen sufficiently large so that all the random walks terminate within this length with high probability. Therefore, it is sufficient to stitch approximately $O\big(\log n/\lambda \eps \big)$ times from each node in parallel. Since each  stitching step can be done in one of round (cf. Lemma \ref{lem:sample-coupon}), the stitching process takes $O\big(\frac{\log n}{\lambda \eps} \big)$ rounds. Now it remains to show the running time of completing the random walks at the end of Phase 2 (Step \ref{step:count} in Algorithm \ref{alg:pr-walk-undirected}). For this step, the length of the random walk  is less than $\lambda$, which are executed in parallel. In this case, we do not need to send any IDs or counters with the coupon, simply send the count of the tokens traversing an edge in a given round to the appropriate neighbors (i.e., in the similar way as of Algorithm \ref{alg:simple-pagerank-walk}). Each token corresponds to a random walk for the remaining length left to complete the length $L$. This will take at most $O(\lambda)$ rounds. Hence, Phase 2 finishes in $O\big(\frac{\log n}{\lambda \eps} + \lambda \big)$ rounds.  
\end{proof}


\begin{lemma}\label{lem:phase3}
Phase 3 finishes in $O\left(\frac{\lambda}{\eps} \right)$ rounds.  
\end{lemma} 
\begin{proof}
Recall that each short walk is of length $\lambda$. Phase 3 is simply tracing back the $\Theta(\log n)$ short walks from each node in parallel. So it is easy to see that we can perform all the reverse walks in parallel in $O(\lambda/\eps)$ rounds (in the same way as to do all
the short walks in parallel in Phase 1). Therefore, in accordance with the Lemma \ref{lem:phase1}, Phase 3 finishes in $O\big(\frac{\lambda}{\eps}\big)$ rounds.  
\end{proof}

Notice that the {\em Coupon IDs} are useful in this context, since the random walks starting at $v$ and ending at $u$ may have followed different paths; $u$ just knowing the number of random walks coming from $v$ is insufficient to backtrace the walks. Moreover, the nodes on the paths will need to know the $Coupon_{ID}$ as well for the same reason. Now we are ready to show the main result of this section. 

\begin{theorem}\label{thm:main}
 The {\sc Improved-PageRank-Algorithm} (cf. Algorithm \ref{alg:pr-walk-undirected}) computes a $\delta$- approximation of the PageRanks with high probability for any constant $\delta$ and finishes in $O\left(\frac{\sqrt{\log n}}{\eps} \right)$ rounds. 
\end{theorem}
\begin{proof}
The algorithm {\sc Improved-PageRank-Algorithm} consists of three phases. We have calculated above the running time of each phase separately. Now we want to compute the overall running time of the algorithm by combining these three phases and by putting appropriate value of parameters. By summing up the running time of all the three phases, we get from Lemmas \ref{lem:phase1}, \ref{lem:phase2}, and \ref{lem:phase3} that the total time taken to finish the {\sc Improved-PageRank-Algorithm} is $O\big(\frac{\lambda}{\eps}+ \frac{\log n}{\lambda \eps} + \lambda + \frac{\lambda}{\eps} \big)$ rounds. Choosing $\lambda = \sqrt{\log n}$, gives the required bound as $O\big(\frac{\sqrt{\log n}}{\eps}\big)$. The correctness and approximation guarantee follows from the previous section.    
\end{proof}

\section{Conclusion}\label{sec:conclusion}
We presented fast distributed algorithms for computing PageRank, a measure of fundamental interest
in networks. Our algorithms are Monte-Carlo and based on the idea of speeding up random walks
in a distributed network. Our faster algorithm takes time only sub-logarithmic in $n$  which can
be useful in large-scale, resource-constrained, distributed networks, where running time is especially crucial. Since they are based on random walks,
which are lightweight, robust, and local, they can be amenable to self-organizing and dynamic networks.

\section*{Acknowledgments}
We thank the anonymous reviewers  for their detailed comments which helped in improving the presentation
of the paper.


  
\bibliographystyle{abbrv}
\bibliography{Distributed-RW}

\end{document}